\documentclass[12pt]{iopart}
\expandafter\let\csname equation*\endcsname\relax
\expandafter\let\csname endequation*\endcsname\relax

\usepackage{amsmath}
\usepackage{amsthm}
\usepackage{iopams}
\usepackage{graphicx}
\usepackage{psfrag}
\usepackage{float}
\usepackage{placeins}
\usepackage{enumitem}

\newtheorem{propos}{Proposition}[section]
\newtheorem{cor}{Corollary}[section]

\eqnobysec 

\begin{document}

\title[Traversable intra-universe wormholes and timeholes]{Traversable intra-universe wormholes and timeholes in General Relativity: two new solutions}

\author{Alexey L. Smirnov}

\address{Institute for Nuclear Research of the Russian Academy of Sciences,
60-th October Anniversary Prospect, 7a, 117312, Moscow, Russia}
\ead{smirnov@ms2.inr.ac.ru}

\begin{abstract}
Using thin shell formalism we construct two solutions of intra-universe wormholes.
The first model is a cosmological analog of the Aichelburg-Schein timehole, while another one
is an intra-universe form of the Bronnikov-Ellis solution.

\end{abstract}


\section{Introduction}
\label{sec:intro}

All traversable wormhole geometries  known so far can be roughly divided by their topological properties into two classes.
The first class describes the so-called inter-universe wormholes. A typical example is a spherically-symmetric
wormhole consisting of two asymptotically flat spacetimes connected by a common throat. 
The throat is the only causal connection between these asymptotically flat universes. 
The second class includes the so-called intra-universe wormholes. For these wormholes the throat causally connects 
two regions of the same universe. In a genuine intra-universe wormhole, these regions can be also connected 
by causal paths which are not threaded through the throat. Embedding diagrams of such wormholes have form of ``handles''. 
Of course, one can effectively turn these wormholes
into their inter-universe siblings by taking large distance between mouths thereby sweeping  
under the rug all effects related to the non-trivial topology.

According to the above classification almost all known wormhole solutions belongs to the first class, and only a small group
conform with requirements for genuine intra-universe wormholes. 
Among them, for example, models constructed in~\cite{AS,ASI,Clement}. Nevertheless, despite the scarcity of exact solutions, 
physical effects in  abstarct intra-universe wormholes were studied~\cite{FN,Thorne0,Thorne1,Thorne2}.  
In particular, as it was shown that causality violations can occur in such spacetimes. 

In the present paper  we are going to describe two more exact solutions of intra-universe wormholes in General Relativity.  
They are inspired by the model proposed in~\cite{AS}. 
This wormhole was constructed from a Reissner-Nordstr\"om~(RN) solution and a specific Majumdar-Papapetrou~(MP) 
solution by using the thin shell formalism~\cite{Israel}.
The MP~solution describes static gravitational field of two charged spheres. Electric and gravitational potentials
are constant on spheres. This crucial feature allows a continuous matching with the~RN spacetime.
In the RN~part black and white holes are causally connected. They constitute a timelike "handle" for the wormhole and necessarily 
leads to causality violations in the resulting spacetime.
Then by joining the RN region with the~MP~spacetime across its two spheres one
turns worldtubes of these spheres into thin shells. 
By construction, transitions through the resulting wormhole are one-way since one must cross horizons.  
Another distinctive feature of this spacetime is that all energy conditions can be satisfied. 
Using the terminology of~\cite{Visser}, such wormholes will be dubbed "timehole". Note in passing, that there exist 
alternative wormhole models based on the system of connected black and white holes.
For instance, the model of inter-universe wormhole proposed in~\cite{KLNS} belongs to this class.

It appears that the approach of~\cite{AS} can be used to construct at least two new solutions.
In particular, in Section~\ref{sec:th} a cosmological extension of the Aichelburg-Shein timehole is constructed.
It becomes possible since there exist cosmological analogs of~MP solutions, the so-called 
Kastor-Traschen~(KT) spacetimes~\cite{KT,Brill1}. Then a KT~solution exists which describes the outer gravitational 
field of two charged spheres comoving in an expanding universe. 
On the other hand, KT~solutions are reduced to
Reissner-Nordstr\"om-de Sitter~(RNdS) spacetimes in the case of spherical symmetry. 

This fact is of importance in Subsection~\ref{sec:junc} where we consider a junction
of the corresponding KT~spacetime with a RNdS wormhole by virtue of thin shells.
This junction gives a cosmological timehole, i.e. it is again a one-way wormhole 
where the null energy condition can be satisfied. 
Moreover, likewise the static timehole, causality violations can also occur in this model.
However at this stage, causal properties of the timehole are mostly obscured by the lack of knowledge about 
causal structure of the~KT~region.  

In Subsection~\ref{sec:kt} a qualitative study of the causal structure of the~KT~region is given. 
Since the corresponding KT~solution is axially symmetric, we try to deduce its causal structure by studying 
null geodesics propagated either along the axis of symmetry or in the middle plane between spheres.
Essentially the same idea was used in~\cite{Brill1} for the two centered~KT spacetime. 
However, the technique of~\cite{Brill1} cannot be directly applied in our case.
We extract the behaviour of geodesics from properties of a specific second order~ODE.    
In particular, it is possible to identify generators of horizon from solutions of this~ODE.
This analysis suggests that shells are enclosed by common inner and white hole horizons during the evolution.
However, cosmological horizon is splitted in the course of expansion thereby causally disconnecting shells.  
  
In Section~\ref{sec:swh} static wormhole geometry supported by a ghost scalar field are considered. 
In the case of  General Relativity minimally coupled to the ghost field, there also exists analog of MP~solutions~\cite{BFSZ}. 
On the other hand, spherically symmetric wormholes exist in the theory~\cite{Bronnikov,Ellis}.
Thus, to build the intra-universe wormhole we can follow essentially the same procedure as in Section~\ref{sec:th}.
Here however, there is an additional subtlety. The thin shell formalism requires the continuity of the ghost across shells. 
Then the wormhole can be constructed in the form of its universal cover. This in turn, allows to treat the ghost filed
as a phase of some complex scalar field.

In Section~\ref{sec:end} some concluding remarks are given.

Appendix~\ref{sec:ap0} contains supplementary results for subsection~\ref{sec:kt} from the theory of second order ODEs.  
They are required for estimation of root numbers in solutions of~\eref{uz}. 

Throughout the paper we use the signature~\mbox{$(-+++)$} and the units where~\mbox{$c=G=1$}.

\section{Cosmological intra-universe timehole}
\label{sec:th}

\subsection{Building the timehole}
\label{sec:junc}
  
As it was already mentioned, the key element of the Aichelburg-Schein timehole~\cite{AS} is an outer spacetime of two charged 
(non-intersecting) spheres $S^+_{1,2}$ of equal radii~$r_0$ which  are held in equilibrium by a balance between 
gravitational attraction and electrostatic repulsion. 
More precisely, the corresponding solution of the Einstein-Maxwell equations is the~MP~spacetime
\begin{equation}
\label{mp}
ds^2 = -V^{-2}dt^2+V^2(dx^2+dy^2+dz^2),
\end{equation}
where $x,y,z$ are cartesian coordinates and $V$ is a solution of Laplace's equation in these coordinates i.e.
\begin{equation}
\label{le}
\nabla^2 V=0
\end{equation}
with Dirichlet boundary conditions 
\begin{equation}
\label{dbc}
V|_{S^+_{1,2}}=V_0=\mbox{const}.
\end{equation}
Therefore continuous junction of spherically-symmetric solutions becomes possible along spheres~$S^+_{1,2}$.

This also suggest an obvious way to get a cosmological analog of the Aichelburg-Schein timehole.
Namely, we can try to replace the ~MP solution~\eref{mp} by the corresponding KT~spacetimes. 
The KT~metrics are solutions of the Einstein-Maxwell equations with a positive cosmological constant~\cite{KT,Brill1}. 
They describe extremely charged masses comoving in a spatially flat expanding or collapsing universe.
Respectively, the RN region of the original model should be replaced by a \mbox{Reissner-Nordstr\"om-de~Sitter}~(RNdS) one.
Hereafter in this subsection, we will label quantities in KT and RNdS regions 
by~$"+"$ and~$"-"$ respectively whenever it is necessary.

Then we are interested in the~KT~solution which gives the outer gravitational field of
two charged spheres~$S^+_{1,2}$ of equal radii~$r_0$. In cosmological coordinates its line element~is 
\begin{subequations}
\label{kt1}
\begin{eqnarray}
ds^2 &=& -\Omega_+^{-2}dt_+^2+\Omega_+^2(dx^2+dy^2+dz^2),\label{kt}\\
\Omega_+ &=& H_+t_+ + V(x,y,z),\quad H_+=\pm\sqrt{\frac{\Lambda_+}{3}},\label{omega0}
\end{eqnarray} 
\end{subequations}
where  $V$ is the solution of \eref{le}, \eref{dbc} we are going to borrow from \cite{AS}.
We will consider only~\mbox{$H_+>0$} branch~of~\eref{kt}, i.e. cosmological expansion. 

Following \cite{AS} we place centers of spheres on the z-axis at $z=\pm d$ and put
\begin{equation}
\label{dbc1}
V_0=\frac{m_+}{r_0}
\end{equation}
in the Dirichlet boundary condition~\eref{dbc}. The parameter~$m_+>0$ is related to the mass of the Kastor-Traschen region.

There are two representations of~$\Omega_+$.
On the one hand, due to the symmetry of the problem, it is convenient to use bispherical (toroidal) coordinates~\mbox{$(-\infty<\mu<\infty,\,0\leq\eta\leq\pi,\,0<\phi\leq 2\pi)$}, see~e.g.~\cite{KK}. Their relation with cartesian coordinates are as follows
\begin{subequations}
\begin{eqnarray}
\label{bsc}
x &=& c\frac{\sin\eta}{\cosh\mu-\cos\eta}\cos\phi,\\
y &=& c\frac{\sin\eta}{\cosh\mu-\cos\eta}\sin\phi,\\
z &=& c\frac{\sinh\mu}{\cosh\mu-\cos\eta},
\end{eqnarray}
\end{subequations}
where~$c$ is yet an arbitrary constant. 
In the bispherical coordinate system, spheres~$S^+_{1,2}$ are given by simple equations~\mbox{$\mu=\pm\mu_0$} and
\begin{eqnarray}
\label{rd}
r_0 &=& \frac{c}{\sinh\mu_0}\nonumber\\
d &=& c\coth\mu_0.
\end{eqnarray}
Then taking into account~\eref{dbc1} one can explicitly write solution~\eref{kt1} as follows
\begin{subequations}
\label{kt2}
\begin{eqnarray}
ds^2 = -\Omega_+^{-2}dt^2_++\Omega_+^2\frac{c^2}{(\cosh\mu-\cos\eta)^2}(d\mu^2+d\eta^2+\sin^2\eta d\phi^2),\label{mpbs}\\
\fl \qquad\Omega_+(t_+,\mu,\eta) = H_+t_++\frac{m_+}{r_0}\left[1+\sum_{n=-\infty}^{+\infty}(-1)^{n+1}\frac{\sqrt{\cosh\mu-\cos\eta}}{\sqrt{\cosh(\mu+2n\mu_0)-\cos\eta}}\right]\label{omega1}.
\end{eqnarray}
\end{subequations}
The potential~$V$ in square brackets was obtained in~\cite{AS} by the method of images and each term in the series is simply
potential of a particular image. Therefore \eref{omega0} can be also written~as
\begin{equation}
\label{omega2}
\Omega_+(t_+,{\mathbf r})=H_+t_++\frac{m_+}{r_0}\left(\sum_{n=1}^{\infty}\frac{(-1)^{n+1}}{\sinh n\mu_0|{\mathbf r} \pm c\coth n\mu_0 {\mathbf k}|}\right),
\end{equation}
where~${\mathbf k}$ is the unit vector along the $z$-axis.

Another part of the cosmological timehole will be a region of the RNdS spacetime 
with~\mbox{$m_-=q_-$}. 
Note that non spherically-symmetric KT solutions are known only in the form \eref{kt}. 
Thus, joining with the RNdS region has to be carried out in cosmological coordinates. 
In these coordinates the RNdS with  \mbox{$m_-=q_-$} is 
\begin{eqnarray}
ds^2 &=& -\Omega_-^{-2}dt_-^2+\Omega_-^2(dr_-^2+r_-^2(d\theta^2+\sin^2\theta d\phi^2)),\nonumber\\
\Omega_- &=& H_-t_- + \frac{m_-}{r_-}\label{rnds},
\end{eqnarray} 
i. e. a spherically-symmetric KT~spacetime.

Cosmological coordinates in \eref{rnds} do not cover the whole spacetime manifold.
It can be seen from the corresponding Carter-Penrose diagrams~\cite{Brill2}. 
In particular,~\Fref{fig:1} shows a complete RNdS manifold when~\mbox{$m=q$} and~\mbox{$mH<1/4$}.   
In this case, an infinite ladder of asymptotically de~Sitter regions exists.
In order to avoid further cluttering the figure we temporarily drop out subscripts~"-".  
\begin{figure}[t]
  \begin{center}
  \psfragscanon
  \psfrag{A}[c][c][0.8][0]{$A$}
  \psfrag{Ap}[c][c][0.8][0]{$A'$}
  \psfrag{App}[c][c][0.8][0]{$A''$}
  \psfrag{B}[c][c][0.8][0]{$B$}
  \psfrag{Bp}[c][c][0.8][0]{$B'$}
  \psfrag{C}[c][c][0.8][0]{$C$}
  \psfrag{Cp}[c][c][0.8][0]{$C'$}
  \psfrag{D}[c][c][0.8][0]{$D$}
  \psfrag{Dp}[c][c][0.8][0]{$D'$}
  \psfrag{E}[c][c][0.8][0]{$E$}
  \psfrag{Ep}[c][c][0.8][0]{$E'$}
  \psfrag{F}[c][c][0.8][0]{$F$}
  \psfrag{Fp}[c][c][0.8][0]{$F'$}
  \psfrag{sli}[c][c][0.7][0]{$t\to\infty$}
  \psfrag{ch}[c][c][0.7][45]{$t=0,\,r\to\infty$}
  \psfrag{ih}[c][c][0.7][-45]{$t\to-\infty,\,r=0$}
  \psfrag{bh}[c][c][0.7][45]{$t\to\infty,\,r=0$}
  \psfrag{r0}[c][c][0.8][0]{$r=r_0$}
  \psfrag{t0}[c][c][0.8][0]{$t=0$}
  \includegraphics[width=0.67\textwidth]{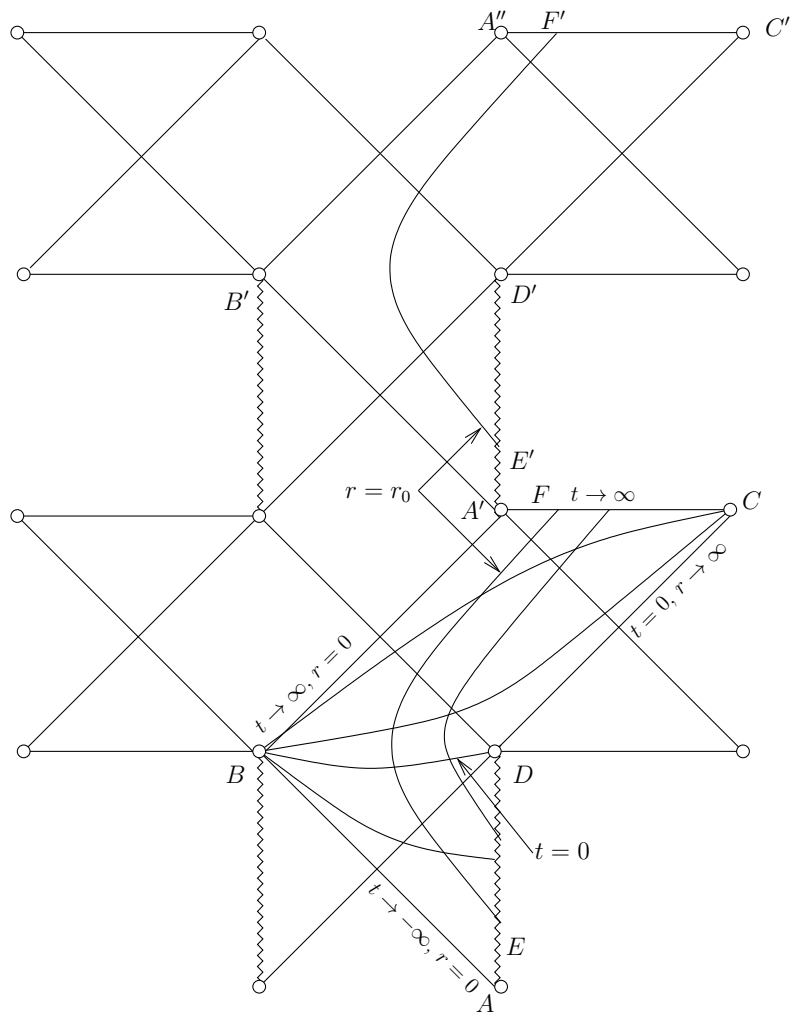}
  \caption{The Carter-Penrose diagram for a RNdS with~\mbox{$mH<1/4$}.
    Arrows point to hypersurfaces~$r=r_0$ and~\mbox{$t=0$}. Seesaw lines are 
    timelike singularities where $\Omega=0$. When matching this spacetime with the KT part to obtain the timehole, 
    regions to the right of hypersurfaces~$r=r_0$ are excised.}
  \label{fig:1}   
  \end{center}   
\end{figure}
Two successive regions covered by different patches of cosmological coordinates lie inside polygons $ABA'CD$ and $A'B'A''C'D'$.
Let us consider the region $ABA'CD$.
In \Fref{fig:1} the left-hand end of the constant~$t$ hypersurfaces is~\mbox{$r=0$}. 
For $t>0$ the right-hand end is~\mbox{$r=\infty$} - the cosmological horizon. For~\mbox{$t<0$} hypersurfaces 
ended at finite~$r$, where the singularity resides and~\mbox{$\Omega=0$}. This means that singularity in these coordinates
is a {\it sphere} at any given time~\mbox{$t<0$}. 
Hence hypersurfaces of constant~$r$ emerge from the singularity when~\mbox{$\Omega=0$}, then
cross consequently the inner horizon, the white hole horizon and the cosmological horizon
and ended at the de Sitter space-like future infinity~${\mathcal J}^+$~\mbox{$(t\to\infty)$}. 

Let us consider hypersurface~\mbox{$r_-=r_0$}. Obviously, it is the worldtube of the sphere~$S^-_1$ 
with {\it coordinate} radius~$r_0$. It bounds the region $ABA'FE$ in~\Fref{fig:1}. 
For observers sitting at sphere~$S^-_1$, their proper time~$\tau_1$ is related with time~$t_-$
in the metric~\eref{rnds} as follows
\begin{equation}
\label{pt1}
\Omega_-|_{r_-=r_0}=\left(H_-t_-+\frac{m_-}{r_0}\right)=e^{H_-\tau_1}.
\end{equation}
It is obvious that~\mbox{$\tau_1=-\infty$} at the singularity and~\mbox{$\tau_1=\infty$} at~${\mathcal J}^+$ .
Then the induced metric on the hypersurface~$r_0$ is
\begin{equation}
\label{mos}
ds^2|_{r_0}=-d\tau_1^2+e^{2H_-\tau_1}r_0(d\theta^2+\sin^2\theta\phi^2).
\end{equation}
This metric shows explicitly that $S^-_1$ is exponentially expanding in its proper time and its areal radius increases.

Further we transform the metric~\eref{kt1} to a spherical coordinate system \mbox{$(r_+,\theta,\phi)$} 
centered at~\mbox{$z=-d$}. Then spheres $S^+_1$ are spacelike sections of hypersurface~\mbox{$r_+=r_0$}. 
We will not distinguish angular coordinates of~"$\pm$" regions.
Now it is possible to continuously identify hypersurfaces~$r_0$ of~$"\pm"$ metrics only if 
\begin{equation}
\label{pt2}
\Omega_+|_{r_+=r_0}=\left(H_+t_++\frac{m_+}{r_0}\right)=e^{H_+\tau_1}
\end{equation}
and cosmological constants~\mbox{$\Lambda_+=\Lambda_-=\Lambda$}. 
Coordinates~$t_+$, $t_-$ differ by a constant shift unless~\mbox{$m_+=m_-$}, as it also follows from~\eref{pt1}, \eref{pt2}.
As result, we obtained thin charged shell $\Sigma_1$ with metric~\eref{mos} on it. 
Similarly,~$\Sigma_2$ is obtained if we identify the woldtube of~$S^-_2$ ($E'F'$ curve) with the woldtube of~$S^+_2$
by the above procedure. 

Applying the Israel equations \cite{Israel,Kuchar}, energy densities~$\sigma_{1,2}$ and stresses~$p_{1,2}$ of 
shells can be expressed by jumps in extrinsic curvatures~$[K_{ab}]$ across shells.
Then taking into account~\eref{mos}, one obtains
\begin{eqnarray}
\label{ieq}
\sigma_i &=& -\frac{1}{4\pi}[K_{\theta}^{\theta}]_{\Sigma_i}\nonumber\\ 
&=&-\frac{e^{-2H\tau_i}}{4\pi} \left(\frac{m_+}{r_0}\frac{\partial V}{\partial r_+}\Big |_{S^+_i}+ \frac{m_-}{r_0^2}\right),\qquad i=1,2,\label{ed}\\
p_i &=& \frac{1}{8\pi}\left([K_0^0]_{\Sigma_i}+ [K_{\theta}^{\theta}]_{\Sigma_i}\right) = 0\label{st},
\end{eqnarray}
where 
\begin{equation}
\label{der}
\frac{\partial V}{\partial r_+}\Big |_{S^+_i}= \mp\frac{\cosh\mu_0-\cos\eta}{4\pi c}\frac{\partial V}{\partial \mu}\Big |_{\mu=\pm \mu_0}.
\end{equation}
are derivatives of $V$ in spherical coordinates centered at~$z=\pm d$.  
Vanishing stresses in~\eref{st} show that shells consist of charged dust. 
Further, according to the~Hopf~lemma for harmonic functions, 
the directional derivative of $V$ along the outward pointing normal to the boundary must be positive, thus
\begin{equation}
\frac{\partial V}{\partial r_+}\Big |_{S^+_i}=-\frac{\partial V}{\partial n}\Big |_{S^+_i}<0.
\end{equation}
Therefore $\sigma_{1,2}$ can be made strictly positive by adjusting $m_+$ and $m_-$, similarly to the static case.

So far our discussion has not touched upon the causal structure of the KT region.
First of all, let us find its behaviour at infinities.
Introduce a spherical coordinate system~\mbox{$(r,\theta,\phi)$} centered at \mbox{$x=y=z=0$}
and consider asymptotic of~\eref{omega2} when~\mbox{$r\to\infty$}
\begin{equation}
\label{omega1as}
\Omega_+(t_+,r,\theta)=Ht_+ +\frac{M_+}{r}
+\frac{Q}{r^3}P_2(\cos\theta) + {\cal O}\left(\frac{1}{r^5}\right).
\end{equation}
Here 
\begin{equation}
\label{Mp}
M_+=2m_+\sum_{n=1}^{\infty}(-1)^{n+1}\frac{\sinh \mu_0}{\sinh n \mu_0}
\end{equation}
is the mass of the KT region.
Alternating series in~\eref{Mp} are absolutely convergent and therefore convergent.
Then similarly to the spherically-symmetric case, 
the region has spacelike infinity~${\cal J}^+$ in the future and is bounded by a cosmological horizon 
in the past (similar to $DC$ line in~\Fref{fig:1}). 
However, if we account for the quadrupole term in~\eref{omega1as},
analytic continuation of the metric across this cosmological horizon fails. 
Some derivatives of the Riemann tensor blow up at the horizon. This is a common feature for all
non spherically-symmetric KT~spacetimes, see~\mbox{\cite[Sec.~IV.B]{Brill1}} for details.

Let us summarise the results. Likewise the model of~\cite{AS}, the obtained timehole is traversable in the sense that
the RNdS region contains causally connected black and white holes as it is seen from~\Fref{fig:1}. 
However, the cosmological timehole is {\it dynamical} since the shells are expanding 
in contrast with the static solution of~\cite{AS}.
Another difference from the static model is the constraint on the type of matter composing the shells. This result is 
in accordance with conclusions obtained in~\cite{Gurses} within alternative approach to multi-shell systems in KT spacetimes.
It was shown in~\cite{Gurses} that two KT spacetimes can be only matched  across dust shells.

Thus the results of this section gives the following still incomplete picture of the cosmological timehole.
Initially only the~RNdS~spacetime exist. Then at time~$t^0_-$ which is defined 
by equation~\mbox{$\Omega_-(t^0_-,r_0)=0$}, shell~$\Sigma_1$ emerges from singularity $AD$ in~\Fref{fig:1}. 
After some time, another shell $\Sigma_2$ emerges from singularity $A'D'$. 
In the~KT~region shells emerge from singularities simultaneously at time~$t^0_+$ 
which is defined from~\mbox{$\Omega_+(t^0_+,r_0)=0$}. 
Such a behaviour suggest that causality violations may  also arise in the model at hand.
Likewise the static timehole, there exist timelike curves from the arbitrary future of~$\Sigma_1$ to
the arbitrary past of~$\Sigma_2$ in the RNdS region. 
Causality violation appears if we are able to causally close these curves through the KT region.
At this stage, it is not clear whether it is possible. 
It also expected that, if exists, the resulting time machine cannot be eternal since the cosmological 
expansions will eventually turn it off by causally disconnecting shells.
These ambiguities require more deeper understanding of the causal structure of the spacetime~\eref{kt1}.

\subsection{Null geodesics in the KT region}
\label{sec:kt}
Information about causal structure of the KT region can be extracted from the behaviour of null geodesics.
Again, we choose coordinates where spheres are centered at~\mbox{$z=\pm d$} and~\mbox{$-z_2<-z_1<0<z_1<z_2$} are respectively crossing points with shells~$\Sigma_1$,~$\Sigma_2$.
Symmetries of the solution suggest that there exist straight light rays propagated either along $z$-axis or in radial directions 
on the middle plane between shells. 

First of all we consider axial null geodesics. 
Then, indeed, by putting \mbox{$x(\lambda)\equiv 0$}, \mbox{$y(\lambda)\equiv 0$} and taking into account that 
\mbox{$\partial_x \Omega(t,0,0,z)=0$} and  \mbox{$\partial_y \Omega(t,0,0,z)=0$} one turns the geodesic equation into a system 
of two ODEs for two functions \mbox{$t(\lambda)$}, \mbox{$z(\lambda)$} of affine parameter~$\lambda$.
On the other hand, line element \eref{kt} can be considered now as the equation for the function~$t(z)$. Namely,
\begin{equation}
\label{ric}
\frac{dt}{dz}=\pm[Ht+V(0,0,z)]^2,
\end{equation}
where $"\pm"$ designates ether outgoing "+" or ingoing "-" rays. Since \eref{ric} is  a Riccati equation it can be 
reduced to a second order linear differential equation. 
One can write 
\begin{equation}
\label{tz}
t=\frac{1}{H}\left(-V\mp \frac{1}{H}\frac{u'}{u}\right),
\end{equation}
where "-" now corresponds to outgoing and "+" to ingoing rays. 
Then \eref{ric} and~\eref{tz} give
\begin{equation}
\label{tzp}
t'=\pm \frac{1}{H^2}\left(\frac{u'}{u}\right)^2.
\end{equation}
The function~$u(z)$ is a particular solution
of the following equation
\begin{equation}
\label{uz}
u''\pm H\frac{\partial V}{\partial z}u =0.
\end{equation}
Here "+" is for outgoing and "-" for ingoing rays.
All signs swap their meanings to opposite when~$z\to-z$. 

Now, using~\eref{tz},~\eref{tzp}, the $t$-component of the geodesic equation becomes identically zero, 
while integration of the $z$-component gives
\begin{equation}
\label{ap}
\lambda =c_1+c_2\int \frac{dz}{u^2(z)}.
\end{equation} 
where $c_1$, $c_2$ are integration constants. 
Thus the function~$u(z)$ defines a map of the affine parameter line into the $z$-axis. 

Further, we can calculate the expansion of the axial null congruence. By using~\eref{tz} and~\eref{ap}, tangential vector field
can be written as
\[
k^{\mu}=\{\pm (\Omega u)^2,0,0,u^2\}.
\]   
A transversal vector field is defined by the condition $k_{\mu}N^{\mu}=-1$ and has the form
\[
N^{\mu}=\frac{1}{2}\{\pm u^{-2},0,0,-(\Omega u)^{-2}\}.
\]
Again "+" is for outgoing and "-" for ingoing geodesics.
Then the expansion is calculated as follows~\cite{Poisson}
\begin{eqnarray}
\label{exp}
\theta &=&\nabla_{\mu}k^{\nu}+\nabla_{\mu\nu}N^{\mu}k^{\nu}+\nabla_{\nu\mu}k^{\nu}N^{\mu}=\frac{\dot Y}{Y},\nonumber\\
     Y &=& \left(\frac{u'}{u}\right)^2=\frac{\dot u^2}{u^6},
\end{eqnarray}
where the overdot is the derivative with respect to $\lambda$ and integration constants are absorbed into $u$. 
The expansion is infinite and~\mbox{$Y(\lambda)=0$} at conjugated points~\cite{HE}. 
However, generators of horizons are special in this respect since they have no conjugate points 
and therefore $\theta(\lambda)$ must be bounded. 

Now, since~\eref{tz} is defined through~\eref{uz}, it is possible to deduce the causal structure directly from
the behaviour of  $u(z)$  by virtue of the following heuristic rules:
\begin{enumerate}[label=(\roman*)]
\item The condition $\Omega=0$ defines singular boundaries for every spacelike hypersurface~\mbox{$t<0$} as
equipotential surfaces of $V$. Then~\mbox{$t=-V(0,0,z)/H$} defines moment of time when the point $z$
emerges from singularity. Evidently, segments of~$u(z)$ will have physical meaning only when 
the corresponding~\mbox{$t(z)\geq -V(0,0,z)/H$} in~\eref{tz}. 
This implies the first rule: {\it $u(z)$, $u'(z)$ must have different signs 
for outgoing and the same signs for ingoing light rays.} 

\item {\it Let $u(z_0)=0$, $u'(z_0)\neq 0$ then intervals $z<z_0$ and $z>z_0$ are causally disconnected by cosmological horizon.}
Indeed, in this case there are two branches of $t(z)$ in the vicinity of $z_0$.
By rule~(i) the physical branch is defined by $t(z_0)\to\infty$.
At the same time,~\eref{ap} and~\eref{exp} suggest that $\lambda\to\infty$ when $z\to z_0$ and 
$\theta(\lambda)$ is bounded for all $\lambda$ and therefore  
such a geodesic is a generator of cosmological (particle) horizon. 
 
\item
{\it A light ray meet a singularity at $z_*$ if $u'(z_*)=0$ and $u(z_*)\neq 0$.}

If $u(z)$ has an extremum at $z_*$ then from~\eref{tz}, we have~\mbox{$Ht_*=Ht(z_*)=-V(0,0,z_*)$} i.e. $\Omega_+(t_*,0,0,z_*)=0$. 
Also, according to~\eref{ap} the affine parameter has a finite value at $z_*$. Thus indeed, the null geodesical incompleteness 
takes place there.
While by construction,~\eref{tz} is regular across $z_*$, rule~(i) defines the corresponding segments of~$u(z)$
unambiguously. 
Let us note however, that if~\mbox{$u''(z_*)=0$} then signs of $u(z)$, $u'(z)$ may remain the same across~$z_*$.  
In this special case, both segments of $t(z)$ may have physical meaning.

\item Let us consider degenerate case $u'(z_*)=0$, $u(z_*)= 0$. By virtue of the L'hopital's rule along with~\eref{uz} one
obtains
\begin{equation}
\label{droot}
\left(\lim_{z\to z_*} \frac{u'}{u}\right)^2=\mp H\frac{\partial V}{\partial z}
\end{equation}
i.e. limit exists if  r.h.s. is non-negative. If r.h.s is zero then there is a singularity at~$(t_*,z_*)$. 
However,~\eref{ap} and~\eref{exp} shows that $\lambda\to\pm\infty$ at~$z_*$
and $\theta$ is bounded for all~$\lambda$. {\it Thus if~$z_*$ is finite, the corresponding geodesic 
is either a generator of past Cauchy horizon
($\lambda\to\infty$) or future Cauchy horizon ($\lambda\to-\infty$) with an end point at~$(t_*,z_*)$.
The hypersurface~$t_*$ has an edge according to~\mbox{\cite[Proposition 6.5.3]{HE}}. 
If~$z_*$~is infinite then the behaviour of generators in the vicinity of point~$D$ at~\Fref{fig:1} is reproduced.}

\end{enumerate}
Note that the potential $V$ in question is regular in its domain. Thus solutions $u(z)$ 
are also regular. For $V$ with poles the above rules should be amended. 
 
At this stage, it is instructive to make a digression and test the rules first for a RNdS spacetime 
in cosmological coordinates. Then
\begin{equation}
\label{rndspot}
\frac{\partial V}{\partial z} = -\frac{m}{z^2}
\end{equation}
and \eref{uz} is the Euler's equation.
Its general solutions are divided into three classes
\begin{equation}
\label{ssktr}
u(z)=
\begin{cases}
|\bar z|^{\frac{1}{2}}(C_1|\bar z|^{\alpha}+C_2|\bar z|^{-\alpha}) & \text{if } \,\,\beta<1 \\
|\bar z|^{\frac{1}{2}}(C_1+C_2\log |\bar z|) & \text{if } |\beta|=1\\
|\bar z|^{\frac{1}{2}}(C_1\sin(\alpha\log |\bar z|)+C_2\cos(\alpha\log |\bar z|))  & \text{if } \,\,\beta>1.
\end{cases}
\end{equation}
Here $\bar z=Hz$, $\alpha=\sqrt{|1-\beta|}/2$ and $\beta=\pm 4mH$ with "-" is for outgoing and "+" for ingoing rays.
It is enough for our purposes to consider only spacetimes with $mH<1/4$ and $mH>1/4$. In the following we will call
such spacetimes subcritical and supercritical respectively.  

In these spacetimes let us choose some hypersurface of constant non-zero radius~$r_b$. 
The hypersurface crosses the $z$-axis at $\pm z_b$ where~\mbox{$r_b=|z_b|$}. 
In order to avoid the pole~\mbox{$r=0$}, let us cut off the spacetime region $r<r_b$. 
We obtained a spherically-symmetric analog of  
the~\mbox{Kastor-Traschen} region of the timehole. We are interested in outgoing and ingoing rays with respect to~$z_b$.
Note, the region $r\geq r_b$ first come into existence when~\mbox{$t_0=-m/Hr_b$}. 
Then rule~(i) implies that all physical solutions~\eref{tz} must have $t(z_b)>t_0$ and
$u(z_b)$, $u'(z_b)$ have different signs for outgoing and the same signs for ingoing rays.  

Outgoing rays are described by the first line of \eref{ssktr} with $1/2<\alpha<1$ in either spacetime. 
One has three allowed shapes of $u(z)$ up to overall sign. They are presented at \Fref{fig:2}.
\begin{figure}[H]
  \begin{center}
  \psfragscanon
  \psfrag{zp}[c][c][0.8][0]{$z_b$}
  \psfrag{zs}[c][c][0.8][0]{$z_*$}
  \psfrag{z0}[c][c][0.8][0]{$z_0$}
  \psfrag{z}[c][c][0.8][0]{$z$}
  \psfrag{u}[c][c][0.8][0]{$u$}
  \psfrag{a}[c][c][0.8][0]{$(a)$}
  \psfrag{b}[c][c][0.8][0]{$(b)$}
  \psfrag{c}[c][c][0.8][0]{$(c)$}
  \includegraphics[width=0.8\textwidth]{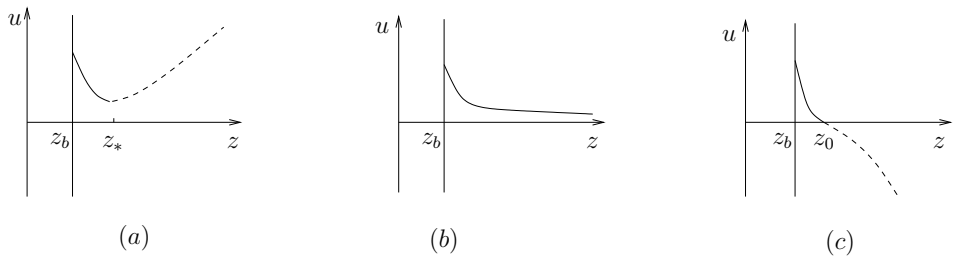}
  \caption{A schematic graphs of $u(z)$ for outgoing rays. Unphysical parts of solutions are depicted by dashed lines.}
  \label{fig:2}   
  \end{center}   
\end{figure}
According to rule~(ii) solutions at \Fref{fig:2}($a$) describes rays emitted from $z_b$ and 
absorbed by singularity. 

Solution at \Fref{fig:2}($b$) has $C_1=0$ and describes a generator of the inner horizon according to rule~(iv). 
Indeed, as one can see from \Fref{fig:1}, generators of the inner horizon have $t(z)\to 0$ when $z\to\infty$. 
Similarly, at \Fref{fig:2}($b$) one has $u(z)\to 0$ and $u'(z)\to 0$ when~$z\to\infty$.  
Therefore as it follows from \eref{droot}, the corresponding~$t(z)\to0$ in this limit. 
This conclusion is confirmed by direct calculation of the radius in static coordinates~\mbox{$R=Hrt+m$} 
by using~\eref{tz} along with~\eref{ssktr}.

At last, by rule~(ii), \Fref{fig:2}(c) describes ray reaching $r=|z_0|=r_0$ in infinite time. 
Hence the hypersurface~\mbox{$r=r_b$} will be eventually separated from $r_0$ by its particle horizon.

For ingoing rays in subcritical case solutions $u(z)$ is still defined by the first line of~\eref{ssktr} with~$\alpha<1/2$.
Possible shapes of $u(z)$ up to overall sign are shown at \Fref{fig:3}
\begin{figure}[H]
  \begin{center}
  \psfragscanon
  \psfrag{zp}[c][c][0.8][0]{$z_b$}
  \psfrag{zs}[c][c][0.8][0]{$z_*$}
  \psfrag{z0}[c][c][0.8][0]{$z_0$}
  \psfrag{z}[c][c][0.8][0]{$z$}
  \psfrag{u}[c][c][0.8][0]{$u$}
  \psfrag{d}[c][c][0.8][0]{$(d)$}
  \psfrag{e}[c][c][0.8][0]{$(e)$}
  \psfrag{f}[c][c][0.8][0]{$(f)$}
  \includegraphics[width=0.8\textwidth]{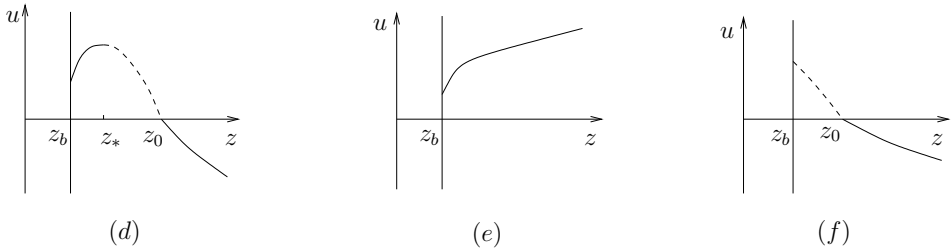}
  \caption{A schematic graphs of $u(z)$ for ingoing rays.}
  \label{fig:3}   
  \end{center}   
\end{figure}
Solutions of type~($d$) contain two physically permitted segments. The positive segment describes rays emerged from singularity. 
The negative segment is in fact equivalent to the type~($f$) and describes generator of a particle horizons.
\Fref{fig:3}($e$) defines rays arriving from the past cosmological horizon when~\mbox{$C_1>0$},~\mbox{$C_2>0$}.
In this case $\lambda$ lies in a finite range and $\theta$ blows up.  
The intermediate case~$C_1=0$ describes a generator of the white hole horizon. According to rule~(iii) 
in this case $\lambda$ increases from~$-\infty$.

For the supercritical case, ingoing rays are defined by the last line of \eref{ssktr}. Solutions are oscillatory.
This is to be expected since the spacetime now contains naked singularity and any ingoing ray must emerge from it.
Then by rule~(iii), infinite number of extrema provides such possibility.  

In order to further clarify the above picture, all these solutions can be mapped into Carter-Penrose diagrams 
as it is shown in \Fref{fig:4}.
\begin{figure}[t]
  \begin{center}
  \psfragscanon
  \psfrag{z0}[c][c][0.8][0]{$r=r_0$}
  \psfrag{zp}[c][c][0.8][0]{$r=r_b$}
  \psfrag{a}[c][c][0.8][0]{$a$}
  \psfrag{b}[c][c][0.8][0]{$b$}
  \psfrag{c}[c][c][0.8][0]{$c$}
  \psfrag{d}[c][c][0.8][0]{$d$}
  \psfrag{e}[c][c][0.8][0]{$e$}
  \psfrag{f}[c][c][0.8][0]{$f$}
  \includegraphics[width=0.8\textwidth]{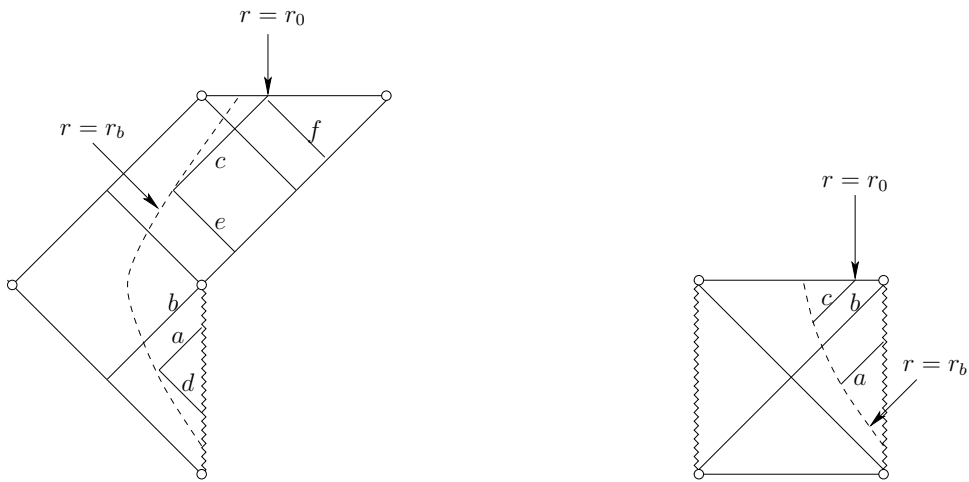}
  \caption{Various types of light rays in Carter-Penrose diagrams.}
  \label{fig:4}   
  \end{center}   
\end{figure}

Thus from the point of view of \eref{uz}, subcritical and supercritical 
cases differ by a number of roots for $u(z)$. In the subcritical case there is at most one root for {\it every} solution $u(z)$.
Hereafter, following \cite{Hartman}, we will say that equation~\eref{uz}
is {\it disconjugate} if every nontrivial solution has at most one zero in its domain. 

It is natural to suppose that the difference by roots numbers is preserved for general KT spacetimes.
This is indeed the case for the metric~\eref{omega2}.
However, in order to not interrupt the flow of the article, technical details were gathered in Appendix.   

Now we proceed with the KT region of the timehole.
Since the Aichelburg-Schein potential $V$ is a harmonic function, 
we know its behaviour along coordinate axes from the maximum principle. 
A schematic shape of $V(0,0,z)$ is presented in \Fref{fig:5}.
\begin{figure}[H]
  \begin{center}
  \psfragscanon
  \psfrag{zm1}[c][c][0.8][0]{$-z_2$}
  \psfrag{zm2}[c][c][0.8][0]{$-z_1$}
  \psfrag{zp1}[c][c][0.8][0]{$z_1$}
  \psfrag{zp2}[c][c][0.8][0]{$z_2$}
  \psfrag{z}[c][c][0.8][0]{$z$}
  \psfrag{V}[c][c][0.8][0]{$V$}
  \psfrag{Vs}[c][c][0.8][0]{$V_s$}
  \psfrag{0}[c][c][0.8][0]{$0$}   
  \includegraphics[width=0.7\textwidth]{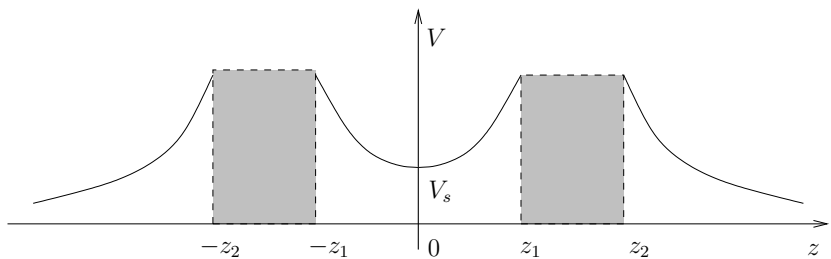}
  \caption{A schematic plot of $V(0,0,z)$. Shaded regions lie inside the spheres.}
  \label{fig:5}   
  \end{center}   
\end{figure}
Also note that the KT region emerges from singularities at~\mbox{$t_0=-m_+/Hr_0$}.
This implies that all physical solutions~\eref{tz} must have $t(z_b)>t_0$.
Therefore on shells,~$u(z)$ and $u'(z)$ must have different signs for outgoing and the same signs for ingoing rays.

In the case $|z|\geq z_2$, it is enough to consider light rays from/to $z_2$ (shell~$\Sigma_2$). 
For outgoing rays, Proposition \ref{prop:2} can be applied. 
Then the general solutions of \eref{uz}  can be written as
\[
u(z)=C_1 u_1(z)+ C_2 u_0(z).
\] 
Moreover, asymptotic for $u_1(z)$ and $u_0(z)$ near infinity according to~\eref{omega1as} are defined 
by corresponding~\eref{ssktr}. 
Using properties of $u_0(z)$ and $u_1(z)$, it is easy to see that allowed shapes of $u(z)$ are still as in \Fref{fig:2}
and have the same meaning.

For ingoing rays, \eref{uz} is disconjugate if~\mbox{$M_+H<1/4$} and $M_+$ is defined by~\eref{Mp}.
The proof is given in Appendix \ref{sec:ap1}. 
Then by using solely~\eref{uz} and the condition~\mbox{$\partial_z V<0$} for $z\geq z_2$, we can reconstruct 
all possible shapes of solutions for~\eref{uz}. They are still described by~\Fref{fig:3} and corresponding
asymptotics~\eref{ssktr}. Therefore, their interpretations remains the same as in the spherically-symmetric case.   

Let us consider the interval \mbox{$|z|\leq z_1$}.
Due to the symmetry $z\to -z$ it is enough to study outgoing and ingoing rays with reference to $-z_1$ (shell $\Sigma_1$)
and $z_1$ (shell $\Sigma_2$) respectively. 
Again, as it is shown in Appendix \ref{sec:ap1}, equation \eref{uz} is disconjugate when ~\mbox{$M_+H<1/4$}.

At this stage, a note about singularities are in order. Singularities evolve  in a general KT spacetime.
We know that they are defined by equipotential surfaces of $V$ at each moment of time. 
For the spherically-symmetric case these are simply spheres of increasing radii when~\mbox{$t\to 0_-$}. 
For the Aichelburg-Schein potential let us draw the shape of $V(x,0,0)$ in addition to~\Fref{fig:5}. It is shown at the next figure
\begin{figure}[H]
  \begin{center}
    \psfragscanon
    \psfrag{x}[c][c][0.8][0]{$x$}
    \psfrag{V}[c][c][0.8][0]{$V$}
    \psfrag{Vs}[c][c][0.8][0]{$V_s$}
    \psfrag{0}[c][c][0.8][0]{$0$}   
    \includegraphics[width=0.7\textwidth]{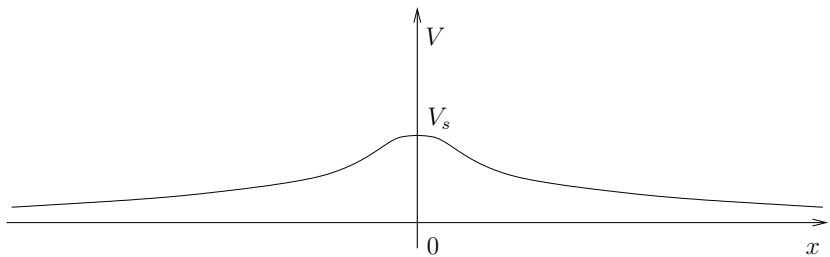}
    \caption{A schematic plot of $V(x,0,0)$.}
    \label{fig:6}   
  \end{center}    
\end{figure}
Then by combining Figures~\ref{fig:5},~\ref{fig:6}, one can see that $V$ is a saddle as it is expected from 
the maximum principle. Dissecting $V$ by equipotential surfaces we arrive to the following picture
\begin{figure}[H]
  \begin{center}
    \psfragscanon
    \psfrag{t1}[c][c][0.8][0]{$Ht<-V_s$}
    \psfrag{t2}[c][c][0.8][0]{$Ht=-V_s$}
    \psfrag{t3}[c][c][0.8][0]{$Ht>-V_s$}
     \psfrag{z}[c][c][0.8][0]{$z$}
    \includegraphics[width=0.8\textwidth]{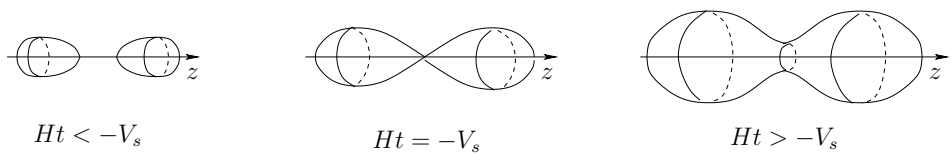}
    \caption{Singular boundaries of the KT region at various moments of time.}
    \label{fig:7}   
  \end{center}    
\end{figure}
Thus we see that initially separated singularities coalesce into a single one when~\mbox{$t=-V_s/H$}.
Moreover, it means the middle plane between shells emerges from the singularity at this moment.

Taking this picture into account we can make a few preliminary observations about possible shapes of $u(z)$.
First of all, singularities are absent on the interval~\mbox{$|z|\leq z_1$} when~\mbox{$t>-V_s/H$}. 
Hence there may exist solutions of \eref{uz} without extrema.
On the other hand, solutions with two extrema are forbidden. Indeed, in this case, one would have
\mbox{$t(-z_*)=t(z_*)=-V(0,0,z_*)/H$} by symmetry, that contradicts \eref{tzp} unless~\mbox{$z_*=0$}.
Also note, that since $z_1$ is the parameter of the model, disconjugate solutions of~\eref{uz} may not acquire   
additional roots when varying $z_1$. 

Having these constraints in mind, we can plot possible disconjugate solutions (up to overall sign)
by using only~\eref{uz}. For outgoing rays from $-z_1$ they are presented at~\Fref{fig:8}.
\begin{figure}[t]
  \begin{center}
    \psfragscanon
    \psfrag{u}[c][c][0.8][0]{$u$}
    \psfrag{zm}[c][c][0.8][0]{$-z_1$}
    \psfrag{zp}[c][c][0.8][0]{$z_1$}
    \psfrag{z0}[c][c][0.8][0]{$z_0$}
    \psfrag{zs}[c][c][0.8][0]{$z_*$}
    \psfrag{a}[c][c][0.8][0]{$(a)$}
    \psfrag{b}[c][c][0.8][0]{$(b)$}
    \psfrag{bp}[c][c][0.8][0]{$(b')$}
    \psfrag{c}[c][c][0.8][0]{$(c)$}
    \psfrag{d}[c][c][0.8][0]{$(d)$}
    \psfrag{e}[c][c][0.8][0]{$(e)$}
    \includegraphics[width=0.8\textwidth]{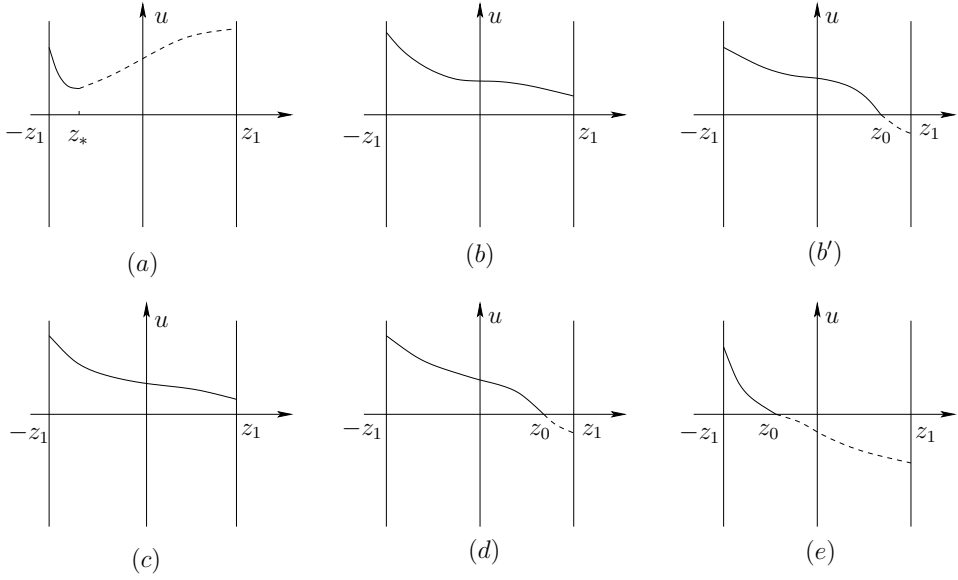}
    \caption{Solutions $u(z)$ on $[-z_1,z_1]$ for light rays emitted from~$z_1$.
      Types $b$ and $b'$ have~$\bar u'(0)=0$.}
    \label{fig:8}   
  \end{center}    
\end{figure}
On the other hand, particular solutions of~\eref{uz} on~$[-z_1,z_1]$, can be written in a number of ways. 
For example, one can consider
\begin{equation}
\label{gsz1z1}
\bar u(z)=
\begin{cases}
 \bar C_1u(z) + \bar C_2u(z)\int_z^0\frac{d\bar z}{u^2(\bar z)}  & \text{if } \,\,z\in [-z_1,0], \\
 \bar C_1u(z) - \bar C_2u(z)\int_0^z\frac{d\bar z}{u^2(\bar z)} & \text{if } \,\,z\in [0,z_1]
\end{cases}
\end{equation} 
where $u(z)>0$ is a solution on $[-z_1,z_1]$ which guaranteed to exist by Proposition~\ref{prop:0}.
Then by imposing constraints on the coefficients, one can obtain solutions depicted in~\Fref{fig:8}.

For instance, solutions~($c$) and~($b$)  requires $\bar u'(0)\leq 0$ and $\bar u(z_1)>0$
if~$\bar C_1>0$,~$\bar C_2>0$. Let us put $\bar C_1=1$ and $\bar C=\bar C_2$ for convenience, then
\begin{equation}
\label{ccc}
u(0)u'(0)\leq\bar C<\frac{1}{\int_0^{z_1}\frac{d\bar z}{u^2(\bar z)}}.
\end{equation} 
Of course, the approach at hand has ambiguities, since behaviour of~$u(z)$ is not known.
If~$u'(0)\leq 0$ then this inequality is automatically satisfied. Otherwise, $z_1$ must be small.
Hence, as expected, the causal picture depends on the distance between shells.    
The left inequality in~\eref{ccc} is strict for the type~($c$), which describes
light rays emitted from shell $\Sigma_1$ and received on $\Sigma_2$. 
If the left inequality is saturated then~($b$) describes special situation when~$\bar u'(0)=0$ and~$\bar u''(0)=0$ 
mentioned in the discussion of rule~(iii) above. 
In this case light ray emitted from ~$\Sigma_1$ absorbed by singularity at~\mbox{$t=-V_s/H$}
and at the same time ray is emitted from the singularity toward~$\Sigma_2$.

A particle horizon is developed between shells when 
inequalities $\bar u'(0)\leq 0$, $\bar u(z_1)<0$ are satisfied. They lead in turn to constraints
\begin{equation}
\label{cdc1}
\bar C \geq  u(0)u'(0)
\end{equation}
and 
\begin{equation}
\label{cdc2}
\bar C > \frac{1}{\int_0^{z_1}\frac{d\bar z}{u^2(\bar z)}}.
\end{equation}
Strict inequalities give types~($d$) and~($e$), while solutions~($b'$) must saturate~\eref{cdc1}.
The resulting equality is compatible with~\eref{cdc2} only if $u'(0)>0$.

Eventually, functions~\eref{tz} generated from solutions of~\Fref{fig:8} can be divided into two sets as it is shown 
at~\Fref{fig:9}. By construction, different $t(z)$'s cannot have intersection points. Therefore 
their positions in the plots are defined unambiguously.  
\begin{figure}[H]
  \begin{center}
    \psfragscanon
    \psfrag{t}[c][c][0.8][0]{$t$}
    \psfrag{z}[c][c][0.8][0]{$z$}
    \psfrag{t0}[c][c][0.8][0]{$t_0$}
    \psfrag{ts}[c][c][0.8][0]{$t_s$}
    \psfrag{a}[c][c][0.8][0]{$a$}
    \psfrag{b}[c][c][0.8][0]{$b$}
    \psfrag{bp}[c][c][0.8][0]{$b'$}
    \psfrag{c}[c][c][0.8][0]{$c$}
    \psfrag{d}[c][c][0.8][0]{$d$}
    \psfrag{e}[c][c][0.8][0]{$e$}
    \includegraphics[width=0.8\textwidth]{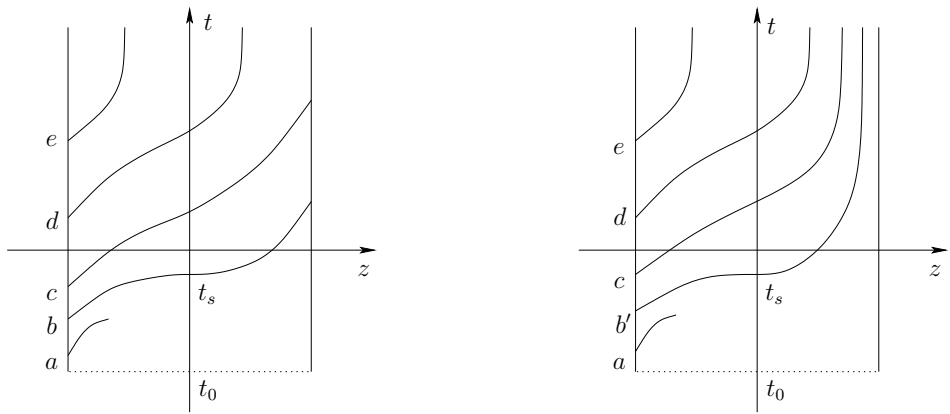}
    \caption{Schematic pictures of outgoing light rays on $[-z_1,z_1]$.}
    \label{fig:9}   
  \end{center}    
\end{figure}
The existence of curves of types ($d$) and ($e$) clearly shows that shells eventually becomes 
causally disconnected along $z$-axis.

Of course,~\Fref{fig:8} and~\Fref{fig:9} represent also ingoing light rays towards~$z_1$.
However, these are not all possibilities. One should also take into account 
light rays which are emitted from the right singularity on $[0,z_1]$  when~\mbox{$t<-V_s/H$}. 
Corresponding~$u(z)$ are depicted at~\Fref{fig:10}
\begin{figure}[H]
  \begin{center}
    \psfragscanon
    \psfrag{u}[c][c][0.8][0]{$u$}
    \psfrag{zm}[c][c][0.8][0]{$-z_1$}
    \psfrag{zp}[c][c][0.8][0]{$z_1$}
    \psfrag{z0}[c][c][0.8][0]{$z_0$}
    \psfrag{zs}[c][c][0.8][0]{$z_*$}
    \psfrag{a}[c][c][0.8][0]{$(a')$}
    \psfrag{b}[c][c][0.8][0]{$(a'')$}
    \includegraphics[width=0.8\textwidth]{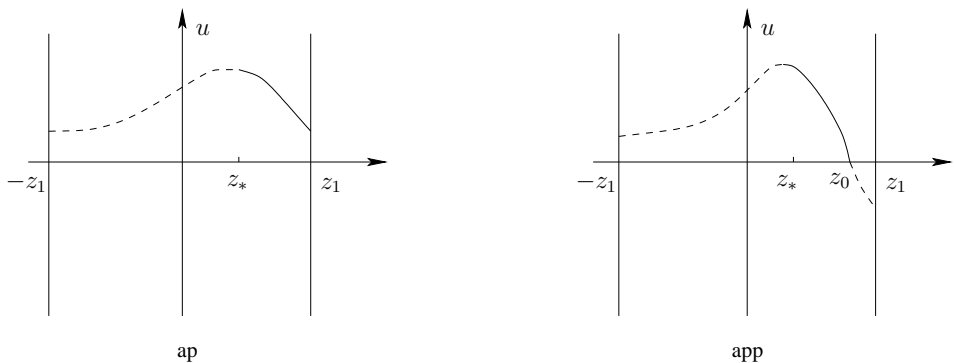}
    \caption{Solutions $u(z)$ for light rays emitted from the right singularity towards~$z_1$.}
    \label{fig:10}   
  \end{center}    
\end{figure} 
In accordance with the above discussion ($a'$) is possible when the distance between shells is small. Otherwise,
the case~($a''$) is realised. 

We note that the disconjugacy forbids solutions with $u(z_*)=0$, $u'(z_*)=0$ which would obey rule~(iv). 
Indeed, by small variations of integration constants $u(z)$ could turn into solution having two roots.

At last, let us briefly discuss null geodesics in the middle plane between shells.
According to~\Fref{fig:7}, the plane emerge from the singularity when~\mbox{$t=-V_s/H$}. 
Further it is convenient to use cylindrical coordinates~\mbox{($z,\rho,\phi)$}. 
Then we are interested in geodesics propagated along $\rho$-directions. 
These geodesics do exist since~\mbox{$\partial_z\Omega(t,\rho,0)=0$} and one can put~\mbox{$z(\lambda)\equiv 0$},~\mbox{$\phi(\lambda)\equiv 0$} in the geodesic equation. Therefore equations~\eref{tz} to~\eref{exp} can be deployed for function~$u(\rho)$,
as well as the rules~\mbox{(i)-(iv)}. Disconjugacy of the corresponding~\eref{uz} is achieved 
for~\mbox{$0\leq\rho<\infty$} when~\mbox{$M_+H<1/4$} and null geodesics are defined similarly to the case~\mbox{$|z|>z_2$} above.
Hence the past Cauchy and the white hole horizons manifest themselves by the presence of the corresponding generators. 

Closing the section, let us summarize our findings.  Behaviour of geodesics along the $z$-axis shows 
no signs of the white hole horizons as well as the past Cauchy horizons between shells. Being combined with 
the results for $\rho$-directions this suggests that there only exist horizons {\it common} for both shells.    

Finally, the above analysis supports assumptions given at the end of the previous section.
Namely, geodesics of type ($c$) in~\Fref{fig:9} will provide causality violation in the complete timehole spacetime provided
the distance between spheres is small enough. However later on, a particle horizon will be developed which will break causal loops.  

\section{Static intra-universe wormhole supported by ghost scalar field.}
\label{sec:swh}
The Aichelburg-Schein potential $V$ described in the previous section can be used to construct intra-universe wormholes
supported by a ghost scalar field. It is possible since there exists solutions somewhat similar to the MP ones 
for the ghost field minimally coupled to gravity~\cite{BFSZ}.
Namely, consider Einstein-ghost theory with the action 
\begin{equation}
\label{gact}
S=\int d^4x \sqrt{-g} \left[\frac{1}{8\pi}R+\nabla_{\mu}\phi\nabla^{\mu}\phi\right],
\end{equation}
According to~\cite{BFSZ}, one has the following static solutions of the Einstein-ghost equations,
\begin{subequations}
\label{bfsz}
\begin{eqnarray}
ds^2 &=& -e^{2V}dt^2+e^{-2V}(dx^2+dy^2+dz^2),\label{m_bfsz}\\
\phi(x,y,z) &=& \frac{V(x,y,z)}{\sqrt{4\pi}},\label{f_bfsz}
\end{eqnarray}
\end{subequations}
where $V$ is the solution of the Dirichlet problem~\eref{le} and~\eref{dbc} which now vanishes at infinity. 
Thus~\eref{bfsz} describes an outer spacetime of two non-intersecting spheres with scalar charges which are held in equilibrium 
by a balance between gravitational attraction and repulsion by the ghost. 

On the other hand, spherically-symmetric static wormhole solutions exist in the theory~\eref{gact}. 
These are Bronnikov-Ellis~\cite{Bronnikov,Ellis} wormhole. 
The technical similarity with the previous section is now evident. 
We can try to construct intra-universe wormhole by matching a Bronnikov-Ellis solution
with~\eref{bfsz} along thin spherical shells.
Let us perform this surgery explicitly.

The Bronnikov-Ellis wormhole is the following spherically-symmetric static solution of theory ~\eref{gact}
\begin{subequations}
\label{be}
\begin{eqnarray}
ds^2 &=& -e^{2v}dt^2+e^{-2v}[dr^2+(r^2+a^2)d\Omega^2],\label{m_be}\\
\phi(r) &=& \frac{v(r)}{\sqrt{4\pi A^2}},\label{f_be}
\end{eqnarray}
\end{subequations} 
Here the radial coordinate $r\in(-\infty,\infty)$ and 
\begin{eqnarray}
v(r) &=& \frac{m}{a}\arctan{\frac{r}{a}},\label{v}\\
A &=& \frac{m^2}{m^2+a^2}\label{a}.
\end{eqnarray} 
The parameter $m$ defines where the wormhole throat resides, 
namely at $r=m$ the minimum of the areal radius~\mbox{$R(r)=e^{-v}\sqrt{r^2+a^2}$} is attained.
Observe also, that opposite sides of the wormhole have nonequivalent characteristics, provided~\mbox{$m\neq 0$}. In particular,
clocks tick at different rate at any~\mbox{$r_2<0<r_1$} since values of the redshift function~\mbox{$e^{v(r_2)}<e^{v(r_1)}$}.
On the other hand, when~$m=0$ we have~$v(r)=0$.
   
Now consider how to join metrics. In the following, solutions~\eref{bfsz} and~\eref{be} along with
related quantities will be labelled respectively by~$"+"$ and~$"-"$. 
In particular, the spacetime~\eref{m_bfsz} is denoted by~$W_+$ and spheres~$S_1^+$,~$S_2^+$ have coordinate radius~$r_0$.
The spheres~$S_1^-$ and~$S_2^-$ are placed respectively at~$r_1>0$ and~$r_2<0$ in~\eref{m_be}, and
the spacetime region~\mbox{$r_2\leq r\leq r_1$} is a 'handle'~$W_-$. 
A straightforward junction of~$W_{\pm}$ similar to the described 
in~\cite{AS} and~Section~\ref{sec:junc} would lead to a continuous metric across static shells~$\Sigma_1$ and~$\Sigma_2$~if
\begin{equation}
\label{rc}
e^{-2V_0}r^2_0 = e^{-2v(r_i)}(r_i^2+a^2),\qquad i=1,2.
\end{equation}
while equating the timelike parts of~$"\pm"$ line elements at these radii.
In the following, in order to proof the concept, we intentionally limit ourselves to the case~$m=0$ and then 
\begin{equation}
\label{r12}
r_1=|r_2|=\sqrt{e^{-2V_0}r^2_0-a^2}.
\end{equation}
We also fix~$V_0$ by~\mbox{$\phi_+|_{\Sigma_2}=\phi_-|_{\Sigma_2}$} and can solve~\eref{r12} for~$r_1$,~$r_2$. 
However, discontinuity in the ghost field at~$\Sigma_1$ is inevitable
\begin{equation}
\label{disc}
D=\phi_-(r_1)-\phi_-(r_2)=2\phi_-(r_1). 
\end{equation}
It leads to failure of the thin shell formalism since the energy-momentum tensor of~$\Sigma_1$
now contains square of the~$\delta$-function.
Therefore to avoid this obstruction, a modification of the model is required. 
It becomes possible if we appeal to the language of covering spaces~\cite{Basener, Geroch}.  

The formal definition of the covering space is the following. Let~$X$ be a topological space and let~$P:\widetilde X\to X$
be an onto map. Then~$P$ is a covering map and~$\widetilde X$ is a covering space for~$X$ if for every~$x\in X$ there is 
a neighbourhood~$U$ of~$x$ such that~$P^{-1}(U)$ is the disjoint union of open sets each of which is mapped homeomorphically 
onto~$U$~by~$P$.

The largest of covering spaces is the universal covering manifold.
Let us choose some fixed point~$x_0\in X$ and let~$C$ be the set
\[
C\equiv\{x,\gamma|x\in X, \gamma\,\,\text{is a curve from}\,\,x\,\,\text{to}\,\,x_0\}.
\]
If~$(x,\gamma)$ and~$(x',\gamma')$ are elements of~$C$ one can write~\mbox{$(x,\gamma)\sim (x',\gamma')$} 
when~$x=x'$ and~$\gamma$ is homotopic to~$\gamma'$. Then~$\sim$ is an equivalence relation in~$C$ and 
the equivalence classes define the points of the universal covering manifold of~$X$.

Returning to our construction, we can generate another Bronnikov-Ellis particular solution by adding an arbitrary 
constant to~\eref{f_be} while leaving the metric unchanged. We choose
\begin{equation}
\label{phipm}
\phi_-'(r') = \phi_-(r) - D.
\end{equation} 
Similarly instead of~\eref{f_bfsz}, a new $"+"$ solution
\begin{equation}
\label{phipp}
\phi_+'(x',y',z')=\phi_+(x,y,z) + D.
\end{equation}
Here primes over coordinates are just trivial relabelling.

Then we continuously join~$W_-$,~$W'_+$ along shell~$\Sigma'_1$ 
and $W_{\pm}$ along~$\Sigma_2$. The junction condition~\eref{rc} on~$\Sigma'_1$ 
equates areal radii of $S^1_-$ and $S^{1'}_+$. 
By construction,~\mbox{$\phi_-|_{\Sigma'_1}=\phi'_+|_{\Sigma'_1}$}, i.e.~$\phi$ is continuous across~$\Sigma'_1$. 
Similarly, the junction of~$W_+$ and~$W'_-$ regions along~$\Sigma_1$ gives~\mbox{$\phi_+|_{\Sigma_1}=\phi'_-|_{\Sigma_1}$}. 
Then we treat the  resulting spacetime as two consecutive sheets~($W'_-\cup W_+$~and~$W_-\cup W'_+$) of a covering spacetime.
Stacking new sheets indefinitely by the above procedure results in a complete covering spacetime~$\widetilde W$ with
everywhere continuous~$\phi$. Since~$\widetilde W$ is simply connected by construction, it is the universal cover
for the intra-universe wormhole.
Geometrically, to obtain the non-simply connected wormhole~$W$ from~$\widetilde W$, we identify all regions~$W_-$,~$W'_-$, etc. 
as well as regions~$W_+$,~$W'_+$,~etc. 
More formally,~$W$ can be written as the factor space~$\widetilde W/\mathbb{Z}$.

Since the corresponding shells are also identified into two mouths~$\Sigma_1$ and~$\Sigma_2$,
their energy densities and stresses must not depend on particular sheet.
The Israel equations~(cf.~\eref{ieq}) give rise to the same result for any sheet
\begin{eqnarray}
4\pi\,\sigma_i &=& e^{V_0}\left(\frac{\partial V}{\partial r_+}\Big|_{r_+=r_0}-\frac{1}{r_0}\right)+\frac{r_i}{r_i^2+a^2},\\
8\pi\,p_i &=& \frac{e^{V_0}}{r_0}-\frac{r_i}{r_i^2+a^2},\quad i=1,2.
\end{eqnarray}
Here, similarly to~\eref{ieq}, derivatives of $V$ are calculated in spherical coordinates centered at one of~$z_+=\pm d$.

Further we can foliate~$W$ by hypersurfaces of simultaneity which 
are homeomorphic to~\mbox{$S^2\times S^1$} minus a point. Indeed, identifying the sheets of~$\widetilde W$ we trivially equate
the respective time coordinates~$t_{\pm}$ and obtain Cauchy surfaces homeomorphic to~\mbox{$S^2\times S^1-\text{pt}$}.

For any pair of sheets in~$\widetilde W$, the difference between field values at the same point is given by some number $nD$, 
where $n\in\mathbb{Z}$. It can be treated as topological invariants of the field~$\phi$. 
Informally, these numbers define how many times field lines wind through the wormhole. 
Since the field is static, we have analog (modulo differentiability) of the winding number for the gradient of~$\phi$. 
More precisely, there exists conserved topological tensor current in any copy of~$W_{\pm}$
\begin{equation}
\label{tcur}
J^{\mu\nu\lambda} = \varepsilon^{\mu\nu\lambda\rho}\partial_{\rho}\phi.
\end{equation}
The corresponding topological charge is a line integral of the Hodge dual of~$J$ i.e.
\begin{equation}
\label{tch}
Q=\int \star J.
\end{equation}
Likewise the winding number it can be used to detect the homotopy class of a path.

Let us consider, for instance, a continuous curve started at the spatial infinity~$i_0$ of~$W$ wound once through 
the throat and ran back to~$i_0$. This curve can be treated as an infinite loop with the base point at~$i_0$. 
The lift of this curve in~$\widetilde W$ can be decomposed into three segments. 
One segment starts at~$\Sigma_2$ then goes through the throat to the next sheet and ends at~$\Sigma'_1$ 
while other two connects~$\Sigma'_1$ and~$\Sigma_2$ with~$i'_0$ and~$i_0$ respectively. 
Then obviously~\eref{tch} gives~\mbox{$Q=-D$} and the loop cannot be shrunk to a point in~$W$. 
The corresponding loop with~$Q=D$ has reverse orientation.
On the other hand, for a shrinkable spacelike loop in the~$"+"$ region with the base point at~$i_0$ we have~$Q=0$. 

These reasonings show that the ghost behaves similar to a phase of a complex scalar field.
In fact, instead of~$\phi$ in~\eref{gact} we can consider a complex scalar ghost with the Lagrangian  
\begin{equation}
\label{csf}
L=-\frac{1}{4\pi^2}\left[\nabla_{\mu}\Phi^*\nabla^{\mu}\Phi +\lambda(\Phi^*\Phi-v^2)\right],
\end{equation}
where~$\lambda$ is a lagrange multiplier.\footnote{Theory~\eref{csf} gives a simple example of automorphic field theory in sense~of~\cite{Dowker}.}
The equations of motions for~$\Phi$
\begin{equation}
\label{emF}
\nabla_{\mu}\nabla^{\mu}\Phi =-\frac{\nabla_{\mu}\Phi^*\nabla^{\mu}\Phi}{v^2}\Phi,\quad\Phi = ve^{i\frac{2\pi}{v}\phi},
\end{equation} 
where~$\phi$ is the original ghost field. Surely~\eref{emF} is trivially reduced to the equation for~$\phi$.  
However Cauchy surfaces of~$W$ is~\mbox{$S^2\times S^1-\text{pt}$}, and~$\Phi$ gives a map from the~$S^1$ factor of this
surface to a circle of radius~$v=D$ in the field complex plane, i.e. it defines the elements of the fundamental group of~$W$. 

At the end of the section let us comment apparent difference in construction of the presented wormhole and the timehole. 
While the wormhole~$W$ is different from its universal cover~$\widetilde W$, in the case of the timehole these spacetimes
coincide. This, in turn, implies that the fundamental group of the timehole manifold must be trivial.

\section{Comments and outlook}
\label{sec:end}
Closing the paper some comments are in order. 

Preliminary analysis shows that obtained spacetimes are unstable. The cosmological timehole inherits 
the Cauchy horizon instability of the~RNdS~region~\cite{BNS}. 
Thus the timehole is traversable modulo this instability. 
The wormhole supported by the ghost seems to inherit instability from its spherically-symmetric parts~\cite{GGS1,GGS2}.
Ideally however, we must consider all spacetimes~$W_{\pm}$ and shells as a coupled system with "periodic" boundary 
conditions~\eref{phipp} imposed on the ghost field. 
There is a dim hope that such a system would be stable in some regime. We leave this question 
to the future work. Still, we consider obtained solutions as useful preliminary step towards constructing more advanced models. 

The fact that the cosmological timehole contains a white hole singularity is expected from the chronology protection conjecture.
Indeed, the absence of singularities would mean that a time machine would develop from regular initial conditions.
On the other hand, for the case $\Lambda=0$, there are theorems~\cite{Tipler} which explicitly forbid such a possibility
when energy conditions are not violated. The cosmological timehole solution suggest that there may exist extensions of 
these theorem to the case of positive cosmological constant.   

Qualitative treatment of the structure of the~KT region given in Subsection~\ref{sec:kt} is still fairly incomplete. 
We were not able to trace generators of horizons in arbitrary directions. 
However the obtained results could be incorporated quite naturally into more general framework
if we use the results of~\cite{Casey}. It was shown in this paper that one can study null geodesics in axially symmetric KT 
spacetimes, projecting them on a plane containing axis of symmetry. The approach of~\cite{Casey} requires to study
a system of third order ODEs for spatial coordinates as functions of the euclidean arc-length parameter~$l$. 
Then the time function~$t(l)$ could be defined through the equation similar to~\eref{uz} but now with respect to~$l$.
It is important that the coefficient before~$u(l)$ in~\eref{uz} is now directional 
derivative~$\dot{\mathbf r}\partial_{{\mathbf r}}V$ and it is always less than steepest descent value~$\partial_zV$. 
Then the difference in root numbers for subcritical and supercritical cases is preserved.
While construction of the equations for apparent horizon is quite nontrivial task the outlined approach could serve as 
its feasible alternative for study of the causal structure of~KT spacetimes.

Results of Section~\ref{sec:swh} can be considered as a preliminary stage towards more advanced set of models.
In particular the obvious next step is to consider wormhole geometries~\eref{be} with~$m\neq 0$.
In this case one could expect causality violations, since the red-shift factor is different for two mouths.
Moreover, we could replace the boundary conditions~\eref{dbc} with more general ones. Namely,   
if solution for the Dirichlet problem with non-equal constant potentials on spheres~$S^+_{1,2}$ could be found, 
it would give us another class of wormholes. It is expected that these solutions would also allow causality violations. 

\section{Acknowledgements}
This work was partially supported by the Russian Foundation for Basic Researches~(RFBR), research project 15-02-05038.
  
\appendix
\section*{Appendix A}
\setcounter{section}{1}
\renewcommand{\thesection}{\Alph{section}}
\subsection{Some results on the theory of disconjugate ODEs.}
\label{sec:ap0}
Let us consider the equation
\begin{equation}
\label{guz}
u''+q(z)u=0
\end{equation}
which is defined on an interval~$I$.

First necessary and sufficient condition for \eref{guz} to be disconjugate is the following
\begin{propos}[Corollary 6.1 \cite{Hartman}]
\label{prop:0}
Let $q(z)$ be continuous on $I$. If $I$ is open or is closed and bounded, then \eref{guz} is disconjugate on $I$
iff \eref{guz} has a solution satisfying $u(z)>0$ on $I$. If $I$ is half-closed interval or a closed half-line, then \eref{guz}
is disconjugate on $I$ iff there exists a solution $u(z)>0$ on the interior of $I$.   
\end{propos} 
Also useful criterion for \eref{uz} to be disconjugate is the so-called "variational principle". A function $\eta(z)$ on subinterval $[a,b]\in I$ belongs to class
$A_1(a,b)$ (or $A_2(a,b)$) if~$\eta(a)=\eta(b)=0$ and $\eta(z)$ is absolutely continuous with $\eta'(z)$ is of class $L^2$ (or~$\eta(z)$ and $\eta'(z)$ are continuously differentiable on~$[a,b]$). Let us define
\begin{equation}
\label{jfun1}
J(\eta;a,b)=\int_a^b(\eta'^2-q\eta^2)dz \qquad \eta\in A_1(a,b).
\end{equation}
If $\eta$ is also in $A_2(a,b)$ then integration by parts gives
\begin{equation}
\label{jfun2}
J(\eta;a,b)=-\int_a^b\eta (\eta''+q\eta)dz \qquad \eta\in A_2(a,b).
\end{equation}
Then the following result holds
\begin{propos}[Theorem 6.2, Exercise 6.3 \cite{Hartman}]
\label{prop:1}
Let $q(z)$ be continuous function on~$I$. Then \eref{guz} is disconjugate on $I$ iff for every closed bounded subinterval 
$[a,b]$ of~$I$, the functional  $J(\eta;a,b)\geq 0$
on class $A_1(a,b)$ (or $A_2(a,b)$) with $J(\eta;a,b)=0$ iff $\eta\equiv 0$. 

If $I$ is not a closed bounded interval then \eref{guz} is disconjugate on it if \eref{jfun2} is satisfied 
for all $[a,b]\in I$ and all $\eta\in A_2(a,b)$.
\end{propos}
By the above proposition, the condition $q(z)\leq 0$ on $I$ is sufficient for \eref{guz} to be disconjugate.
Then the following result is valid 
\begin{propos}[Corollary 6.4 \cite{Hartman}]
\label{prop:2}
Let $q(z)\leq 0$ be continuous on $I:a\leq z<b$. Then \eref{guz} has a solution satisfying 
\[
u_0(z)>0,\qquad u_0'(z)\leq 0 \qquad \text{for} \qquad a\leq z<b 
\]
and a solution $u_1(z)$ such that 
\[
u_1(z)>0,\qquad u_1'(z)\geq 0 \qquad \mbox{for} \qquad a\leq z<b. 
\]
\end{propos}
Here, it is important to note that $u_0(z)$ and $u_1(z)$ are linearly independent by construction.

Another useful criterion is 
\begin{propos}[Corollary 7.1 \cite{Hartman}]
\label{prop:3}
Let $q(z)$ be continuous function on \mbox{$I:a\leq z<b$}, $C$ is a constant , and 
\begin{equation}
\label{Q}
Q(z)=C-\int_a^z q(s)ds.
\end{equation}
If the differential equation 
\begin{equation}
\label{Quz}
u''+4Q^2(z)u=0
\end{equation}
is disconjugate on $I$ then \eref{guz} is also disconjugate on $I$.
\end{propos}

\subsection{Disconjugacy of equation \eref{uz}.}
\label{sec:ap1}
\begin{cor}
If~\mbox{$M_+H<1/4$} where~$M_+$ is defined by~\eref{Mp} then the equation~\eref{uz} is disconjugate on  $z_2\leq z$. 
\end{cor}
\begin{proof}
For ingoing rays, we apply Proposition \ref{prop:3}.  Then \eref{Q} gives 
\begin{equation}
\label{Qz}
Q(z)=HV(0,0,z)
\end{equation}
by properly adjusting the constant $C$. On the other hand, let a continuous function~$Q^+(z)\geq Q(z)$ on~$z\geq z_2$
then corresponding functionals \eref{jfun2}
\begin{equation}
\label{Jp}
J^+(\eta;a,b)\leq J(\eta;a,b)
\end{equation}  
for all subintervals $[a,b]$.
Therefore, if $J^+(\eta;a,b)\geq 0$ then \eref{Quz} is disconjugate by Proposition \ref{prop:1}. 
It can be shown that
\begin{equation}
\label{Qpz}
Q(z)<Q^+(z)=\frac{M_+H}{z-d}
\end{equation}
where~$d$ as in~\eref{rd}. 
The equation of type \eref{Quz} for~$Q^+(z)$ is the Euler's equation.
It follows from \eref{ssktr} that there exists a solution~\mbox{$u(z)>0$} on~$[z_2,\infty)$ if~$M_+H<1/4$. 
Therefore let us rewrite $J^+(\eta;a,b)$ by using function~\mbox{$\zeta(z)=\eta(z)/u(z)$}. 
Then integration by parts yields 
\[
J^+(\eta;a,b)=\int_a^b u^2\zeta'^2 dz >0
\]
cf. also proof of Theorem 6.2 in \cite{Hartman}.
Hence, by Proposition \ref{prop:3}, equation \eref{uz} is disconjugate for ingoing rays if~\mbox{$M_+H<1/4$}.
\end{proof}
 
Note, this inequality is also sufficient for~\eref{uz} to be \mbox{non-oscillatory} at infinity.
In this weaker case every solution of~\eref{uz} has at most a finite number of roots in its domain. 
This result is provided by~\mbox{\cite[Theorem 7.1]{Hartman}} and asymptotic~\eref{omega1as}.
\begin{cor}
If~\mbox{$M_+H<1/4$} the equation~\eref{uz} is disconjugate on~ $|z|\leq z_1$.
\end{cor}
\begin{proof}
In this case Proposition \ref{prop:3} is not useful since it is defined on a half-closed interval.
Instead, the function
\[
q(z)=\pm H\frac{\partial V}{\partial z}
\]
on~\mbox{$I=[-z_1,z_1]$} can be directly majorized by
\begin{equation}
\label{qpz}
q^+(z)=
\begin{cases}
\frac{HM_+}{(c-z)^2} &\qquad \text{for outgoing rays},\\
\frac{HM_+}{(c+z)^2} &\qquad \text{for ingoing rays}.
\end{cases}
\end{equation} 
Note, that $z_1<c$ and \eref{qpz} are bounded on the interval in question.
Then~\eref{Jp} is satisfied for any subinterval~\mbox{$[a,b]\in I$}. For \eref{qpz}, equation \eref{guz}
has a solution~\mbox{$u(z)>0$} on $I$ provided~\mbox{$M_+H<1/4$}. 
Therefore this equation is disconjugate by Proposition \ref{prop:0}. 
From Proposition~\ref{prop:1} it follows that $J^+(\eta;a,b)$ is positive.
Thus~$J(\eta;a,b)$ is also positive and~\eref{uz} is disconjugate when~\mbox{$M_+H<1/4$}.
\end{proof}

\end{document}